\documentclass[final,nomarks]{dmtcs-episciences}

\usepackage[utf8]{inputenc}
\usepackage{subfigure}

\usepackage{amssymb,amsmath,amsthm}
\usepackage{nicefrac}
\usepackage{url}
\usepackage[normalem]{ulem} 
\usepackage{color}
\usepackage[american]{babel}
\usepackage{version}
\usepackage{thm-restate}
\usepackage{enumerate}
\usepackage[round]{natbib}
\newcommand{\doi}[1]{doi: \href{https://doi.org/#1}{\nolinkurl{#1}}}

\graphicspath{{figures/}}

\newtheorem{theorem}{Theorem}
\newtheorem{lemma}[theorem]{Lemma}
\newtheorem{definition}[theorem]{Definition}
\newtheorem{corollary}[theorem]{Corollary}
\newtheorem{proposition}[theorem]{Proposition}
\newtheorem{observation}{Observation}

\iffalse
\newcommand{\deleted}[1]{\sout{#1}}
\newcommand{\delete}[1]{\sout{#1}}
\newcommand{\changed}[1]{\textcolor{red!70!black}{#1}}
\else
\newcommand{\deleted}[1]{}
\newcommand{\delete}[1]{}
\newcommand{\changed}[1]{#1}
\fi

\usepackage[round]{natbib}

\author[T. Biedl]{Therese Biedl\affiliationmark{1}\thanks{%
This research was supported by NSERC, FRN RGPIN-2020-03958.
The author would like to thank Daniel Rutschmann for inspiring
discussions, and Dinis Vitorino for careful proofreading.
}
}
\title[Face-hitting Dominating Sets]{Face-hitting dominating sets in plane graphs: Alternative proof and  linear-time algorithm}
\affiliation{%
  University of Waterloo, Waterloo, Ontario, Canada}
\begin{document}
\publicationdata{vol. 28:4, SOFSEM 2026}{2026}{3}{10.46298/dmtcs.17747}{2026-03-17; 2026-03-17; 2026-06-18}{2026-07-01}
\maketitle

\begin{abstract}
In a recent paper, Francis, Illickan, Jose and Rajendraprasad showed
that every plane graph $G$ has (under some natural restrictions)
a vertex-partition into two sets $V_1$ and $V_2$ such that each $V_i$
is \emph{dominating} (every vertex of $G$ contains a vertex of $V_i$ in its
closed neighbourhood) and \emph{face-hitting} (every
face of $G$ is incident to a vertex of $V_i$).   Their proof works by
considering a plane supergraph $G'$ of $G$ that has certain properties, and
among all such graphs, taking one that has the fewest edges.   As such,
their proof is not obviously algorithmic.   Their proof also relies on the
4-colour theorem, for which the fastest known algorithm is quadratic (and not simple).

In this paper, we give a new proof that every plane graph $G$ has
(under the same restrictions) a vertex-partition into two 
face-hitting dominating sets.   Our proof is constructive, and requires nothing more
complicated than splitting a graph into 2-connected components, finding
an ear decomposition, and computing a perfect matching in a 3-regular
2-connected plane graph.   For all these problems, linear-time algorithms are known
and so we can find the vertex-partition in linear time.

A crucial ingredient in our proof is to find (for a 2-connected plane
graph $G$) a plane supergraph $G^+$ that is a triangulation, possibly with parallel
edges.    Furthermore, for every vertex $v$ (with one exception) 
the cyclic order of edges around $v$ in $G^+$ contains two consecutive
edges that are in $G$.  
We study this problem further, and show
that except for the case where $G$ is an odd cycle, we can actually find
a triangulation where \emph{every} vertex has such a pair of edges.
\end{abstract}

\section{Introduction}
\label{sec:intro}

One of the oldest problems in graph theory is how to find a small
\emph{dominating set} in a graph $G$, i.e., a set $S$ of vertices
such that for every vertex $v$, either $v$ is in $S$, or some neighbour
of $v$ is in $S$.     The goal is to minimize $|S|$, and the corresponding
decision problem is NP-complete and remains NP-complete even in planar graphs 
under restrictions
such as being 4-regular or having maximum degree~3 \citep{GJ79}.
However, in a planar graph one can efficiently approximate a minimum dominating set
arbitrarily closely, using Baker's approximation scheme \citep{Baker94}.

Related to finding a minimum dominating set is the question how big it
may be required to be.
Since a graph consisting of $n$ isolated vertices clearly requires
size $n$ for any dominating set, one typically imposes some restriction on the
class of graphs studied.   One interesting result here was given by 
\citet{MT96}, who showed that every $n$-vertex simple triangulated planar graph $G$ 
has a dominating set of size at most $n/3$.   
(See Section~\ref{sec:preliminaries} for definitions.)
In fact, they show the stronger result that the vertices of $G$ can be partitioned into 
three disjoint sets that each are dominating.    By combining disjoint copies of $K_4$,
one can easily construct planar triangulated graphs
that require at least $n/4$ vertices in any dominating set, and the quest to bring the upper bound
closer to $n/4$ has given rise to a number of interesting research insights. 
See the paper by \citet*{CRR24} for the currently 
best upper bound of $\tfrac{2}{7}n$ (if $n>10$), and the
references therein for more results for other related classes of planar graphs.

In a recent paper, \citet*{FIJR24}
studied an interesting variant of this problem, where they are looking for
a vertex set that is not only a dominating set, but that also is 
\emph{face-hitting}, i.e.,  they presume that the planar graph $G$ comes with a
fixed crossing-free drawing (it is \emph{plane}),
and they want a dominating set $S$ where
additionally every face of the drawing is incident to at least one vertex of $S$.
It is known from an earlier construction by \citet*{BSTZ97}
that there are simple triangulated planar  graphs where every face-hitting vertex set has
size at least $\lfloor n/2 \rfloor$.    (In a triangulated planar graph, such a set
is automatically also dominating.) Francis et al.~match this bound, 
and in particular, show
that any plane graph (under some natural restrictions, such as ``no isolated vertices'')
has a vertex-partition into two face-hitting dominating sets; the smaller
of these sets has size at most $\lfloor n/2 \rfloor$.   
\changed{(See also Figure~\ref{fig:intro}(a).)}
This in turn can be used to improve the bound of \citet{CRR24}
in some triangulated planar graphs: 
If $G$ has an independent set $I$ of size exceeding $3n/7$,
then $G\setminus I$ has fewer than $4n/7$ vertices, and
by taking a small face-hitting dominating set of $G\setminus I$ one can obtain a 
dominating set of $G$ of size less than $2n/7$.

From a mathematical point of view the result of \citeauthor{FIJR24} is very
nice since it matches the lower bound. But from an algorithmic point of view it is somewhat
disappointing, since the authors do not discuss how to find the vertex-partition efficiently,
and it is not obvious how to do so.   
In particular, a crucial ingredient of their proof is to use a plane supergraph $G'$ of
$G$ that has the maximum number of so-called ``happy vertices'' and ``happy faces'',
and among all such
graphs, take the one with the fewest added edges.    Clearly $G'$ always
exists, but it is not clear how to find it.   From later properties that the authors show
about $G'$, it seems likely that one could simply start with $G':=G$, search for
a violation of those properties, and then update $G'$ according to the steps in their
proofs.   It is not clear how fast such a procedure would be, but it would likely take at
least quadratic time.

A second bottleneck is that Francis et al.~use the 4-colour theorem to colour the graph $G'$ (and
extract the vertex-partition by combining two colour-classes).   While there is an
algorithm to find the 4-colouring of a planar graph by \citet*{RSST97}, the run-time is quadratic
(with a very large constant factor). 

\begin{figure}[ht]
\hspace*{\fill}
	\subfigure[]{\includegraphics[scale=1.1,page=3]{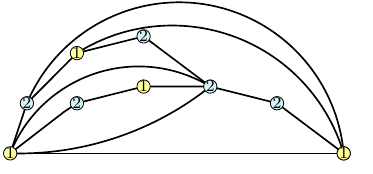}}
\hspace*{\fill}
	\subfigure[]{\includegraphics[scale=1.1,page=4]{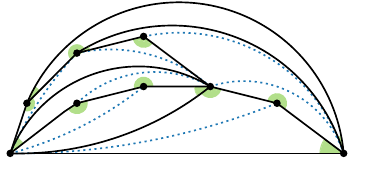}}
\hspace*{\fill}
\changed{
	\caption{(a) A plane graph $G$, not simple, with a partition into two face-hitting dominating sets.   
	(b) A supergraph $G^+$ (added edges are blue dotted). Graph $G^+$ is a triangulation where every vertex is happy due to the green shaded angles.}
}
\label{fig:intro}
\end{figure}

In this paper, we therefore re-prove the result by Francis et al.~in a different way
that immediately leads to a linear-time
algorithm to find the vertex-partition.    Rather than relying on the 4-colour
theorem, we find a bipartite subgraph $H$ of $G$ whose 2-colouring then gives the
vertex-partition.   We ensure that for every vertex and every face of $G$ 
at least one incident edge is retained in $H$, so either colour-class is a face-hitting
dominating set.     We sketch here an outline of how to find $H$.   This 
is easy if $G$ is triangulated: remove the edges corresponding to a perfect matching in the dual of $G$ (this always
exists by Petersen's theorem).   If $G$ is not triangulated but 2-connected, then we first create a triangulated 
plane supergraph $G^+$ of $G$ where every vertex $v$ (with one exception) is \emph{happy}:
in the cyclic order of edges around $v$ in $G^+$ there are two consecutive
edges that are in $G$.  \changed{(See also Figure~\ref{fig:intro}(b).)} Graph $G^+$ is obtained by computing an ear decomposition of $G$, and 
triangulating each face suitably when its corresponding ear is added.   Since
vertices of $G^+$ are happy, removing the edges of a perfect matching in the dual of $G^+$
then again gives $H$.   Finally, for a graph $G$ that is not 2-connected,
we simply combine the subgraphs of the 2-connected components.  All steps can easily be 
implemented in linear time, leading us to our main result:

\begin{theorem}
\label{thm:main}
There exists a linear-time algorithm for the following task:   Given a plane
graph $G=(V,E)$ without isolated vertices or faces incident
to at most two vertices, partition $V$ into two vertex sets $V_1,V_2$ that 
both are dominating as well as face-hitting.
\end{theorem}

The main new ingredient in our paper is an algorithm that finds, for a given 2-connected plane graph $G$, 
a triangulated supergraph $G^+$ with many happy vertices.
This problem is related to other papers
that dealt with triangulating a planar graph under some constraint, see e.g.~\citet{KB97} for
triangulating planar graphs while keeping the maximum degree small, and \citet{BKK97} for
triangulating planar graphs while making them 4-connected.     
There are also many papers that consider triangulating geometric objects (such as a set of points or a
polygon) under some constraints, see e.g.~\citet{BEEMT93} and the references therein.
Finally, making vertices happy can be seen as the problem of assigning angles to 
vertices under some constraints, and as such our problem 
is also related to papers that studied angle-assignments (under different constraints), see
e.g.~\citet{EGSW21}.
But to our knowledge the question of triangulating a planar graph while making vertices happy
has not been studied previously except by \citet{FIJR24}.    
We feel that it deserves more attention, and therefore study
it beyond what is needed as ingredient for the proof of Theorem~\ref{thm:main}.   Specifically,
we show the following.

\begin{restatable}{theorem}{ThmMakeHappy}
\label{thm:make_happy}
\label{thm:makeHappy}
	Let $G$\delete{$=(V,E)$} be a 2-connected plane graph with \changed{at least three} vertices.   If $G$ is not an odd-length cycle,
then there exists a (possibly non-simple) triangulated plane supergraph $G^+$ without loops
such that all vertices are happy.
It can be found in linear time.
\end{restatable}

We will also show that the restrictions ``2-connected'', ``not an odd-length cycle'' and
``possibly non-simple'' are required in the sense that for some graphs, not all
vertices can be made happy otherwise.

\section{Preliminaries}
\label{sec:preliminaries}

This section clarifies some notation for a graph $G=(V,E)$;
for most standard notation, see for example Diestel's book \citep{Die12}.
Throughout, $n$ is the number of vertices, and we assume \changed{that} $n\geq 2$.
A \emph{loop} is
an edge where the two endpoints coincide, while two \emph{parallel edges}
are two edges with the same set of endpoints.   A graph is \emph{simple} if
it has no loops and no parallel edges.    We do not assume that $G$ is simple
(except where explicitly stated below).
A \emph{$k$-walk} is a closed walk with $k$ edges; it is called a
\emph{$k$-cycle} if no vertex repeats.

Recall that a vertex set $S$ is \emph{dominating} if for every $v\in V$
either $v\in S$ or some neighbour of $v$ is in $S$.   An \emph{edge cover}
is a set $S$ of edges such that for every $v\in V$ at least one incident edge
is in $S$.   A \emph{perfect matching} is a set $M$ of edges where every
vertex is incident to exactly one edge in $M$.  

Throughout the paper, we will assume that $G$ is connected, since otherwise 
we can obtain a partition into face-hitting dominating sets in each component and
combine them to get one for $G$.   
Graph $G$ is called \emph{2-connected} if it does not contain a \emph{cut-vertex},
i.e., a vertex $v$ such that $G\setminus v$ has more connected components than $G$.
A \emph{2-connected component} $C$ is a maximal subgraph of $G$ that is 2-connected.
For the borderline case of a vertex $v$ with incident loop(s) in a graph with \changed{at least two}
vertices, we consider
$v$ to be a cut-vertex and each subgraph formed by taking $v$ and one of the loops 
to be a \emph{trivial} 2-connected component.  

A \emph{planar graph} is a graph $G$ that can be drawn without crossing in
the plane, or equivalently, on the sphere $\mathbb{S}^2$.
Consider a crossing-free drawing $\Gamma$ of $G$ on the sphere. 
The maximal regions of $\mathbb{S}^2\setminus \Gamma$ are called the \emph{faces}. 
Since $G$ is connected, drawing $\Gamma$ can be described by specifying 
the \emph{rotation scheme}, i.e., the clockwise cyclic order of edges at each vertex.
In this paper, the input graph $G$ will always be planar, and in fact it will
be \emph{plane}, i.e., it comes with a fixed rotation scheme.
This rotation scheme determines the faces implicitly: 
Every face can be described via the walk that defines its \emph{facial boundary}, and
this walk can be read from the rotation scheme.
An \emph{angle $\angle uvw$ incident to vertex $v$}
consists of two edges $(u,v)$ and $(v,w)$ that are consecutive in the cyclic order
of edges at $v$.
If $\deg(v)=1$ then $v$ has exactly one incident angle (using the same edge twice);
if $\deg(v)\geq 2$ then at all incident angles the two edges are distinct.
For every angle the two edges are consecutive on the facial boundary of
some face $F$; we hence also call it an \emph{angle incident to $F$}.

The \emph{degree} of a face $F$ is the number of angles incident to it, while the \emph{size} of $F$
is the number of vertices incident to $F$.   For a 2-connected plane graph
the two concepts coincide since the boundary of $F$ is a cycle.
A \emph{bigon} is a face whose boundary consists of two parallel edges and so has
degree 2.  
A \emph{triangulated graph} is a plane graph where every face has degree 3, i.e., it is triangular.
Note that this definition forbids bigons and faces bounded by loops, 
but in this paper we specifically \emph{allow} a triangulated
graph to have parallel edges or loops as long as they do not bound a face.   
(We use the term ``simple triangulated graph'' if no loops or parallel edges are allowed.)
A \emph{$3^+$-face} is a face with size 3 or more; this obviously excludes faces of degree~1 or~2,
but it also excludes some faces of higher degrees, for example a triangular face that is bounded by a loop and two parallel edges.  

Let $S$ be either a vertex-set or an edge-set.  A face $F$ is \emph{hit} by $S$ 
if it is incident to at least one element of $S$.   Set $S$ is called
\emph{face-hitting} if it hits every face, and \emph{$3^+$-face hitting} if it
hits every $3^+$-face.   

In this paper, we are usually given a plane graph $G$, and we will work with some
subgraph or planar supergraph of $G$ obtained by deleting/adding edges.
Without further mentioning we assume that these graphs use the rotation schemes 
obtained as explained below.   For a
subgraph $H$ of $G$, we use the \emph{inherited} rotation scheme, i.e., we keep
the same cyclic order of edges at all vertices except for deleting edges that are
not in $H$.    Any planar supergraph $G^+$ that we construct is 
a \emph{plane supergraph}, i.e., it \emph{respects the embedding of $G$}.   This
means that for any newly added edge $e=(v,w)$ there is a face $F$ incident to both $v$ and $w$,
and we add the edge \emph{inside face $F$}, which means that edge $e$ gets inserted
into the cyclic orders at $v$ and $w$ between the two edges of an angle that
is incident to $F$. (We will do this only if $G$ is 2-connected and the angles
are hence unique.)

The \emph{dual graph} of $G$ is the graph $G^*=(\mathcal{F},E^*)$ that has a vertex
$v_F$ for every face $F$ of $G$, and a \emph{dual edge} $e^*=(v_F,v_{F'})$ whenever
faces $F$ and $F'$ share an edge $e$ on their boundary; it includes a loop at
vertex $v_F$ whenever face $F$ is incident twice to an edge $e$.   The degree
of vertex $v_F$ in $G^*$ equals the degree of face $F$ in $G$.
The dual graph $G^*$
has a natural rotation scheme:   For every vertex $v_F$ of the dual
graph, list the incident edges in the order in which their dual edges appeared 
around the facial boundary of face $F$.

It will sometimes be convenient to \emph{choose an outer face of $G$} as follows.
We have a spherical drawing $\Gamma$ of $G$ (this can be read from the rotation scheme).
For any face $F$ of our choice, we can pick a point $p$ strictly inside $F$ and
project the drawing $\Gamma$ into the plane using $p$ as projection point.
With this, every cycle of $G$ then has an interior and exterior side.
Also face $F$ gets projected to an unbounded area of the plane, and we hence
call it the \emph{outer face} 
while all other faces are \emph{interior}.  An \emph{interior vertex} is a vertex 
not incident to the outer face.

\section{Bipartite Edge Covers}

To prove Theorem~\ref{thm:main}, it suffices to prove the following result.

\begin{restatable}{theorem}{BipEdgeCover}
\label{thm:bipEdgeCover}
	Let $G$\deleted{$=(V,E)$} be a connected plane graph with \changed{at least two} vertices.
	Then $G$ has a bipartite subgraph $H$ \deleted{of $G$} whose edges form a $3^+$-face-hitting
edge cover of $G$.  Furthermore, we can pre-specify a \emph{reference-edge} $\hat{e}$ of $G$ and require
that $H$ includes $\hat{e}$.   Subgraph $H$ can be found in linear time.
\end{restatable}

To see why this implies Theorem~\ref{thm:main}, consider a 2-colouring of $H$, say into vertex-sets
$V_1$ and $V_2$,   and let $i\in \{1,2\}$. Since $E(H)$ is an edge cover, for every vertex
$v$ there exists an edge $(v,w)$ in $H$.  Since $H$ is 2-coloured, either $v\in V_i$ or $w\in V_i$, and so $v$ is dominated by $V_i$. 
Since $E(H)$ is $3^+$-face-hitting, for every $3^+$-face $F$ of $G$ there exists an edge $(x,y)$ incident to $F$ in $H$;
either $x\in V_i$ or $y\in V_i$, and so $V_i$ is $3^+$-face-hitting.   Since by assumption of Theorem~\ref{thm:main}
all faces of $G$ are $3^+$-faces, the result holds.   (We note here
that Theorem~\ref{thm:main}, which we intentionally kept very similar to the statement in \citet{FIJR24},
could be strengthened ever so slightly to permit faces of size 1 or 2, as long as they need not be hit.)

\subsection{Triangulated graphs}

We first prove Theorem~\ref{thm:bipEdgeCover} for a triangulated planar graph without loops, where this is
very easy.   For later use as a subroutine, we actually prove a stronger claim on the edges $E(H)$ of $H$: at
every angle at least one edge is retained.
This immediately implies that $E(H)$ is a face-hitting edge cover of $G$ since by connectivity
every vertex and every face is incident 
to an angle.

\begin{lemma}
\label{lem:triangulated}
	Let $G$ be a triangulated planar graph with \changed{at least three} vertices and no loops. Fix a refe\-rence-edge~$\hat{e}$.   Then $G$
has a bipartite subgraph $H$ that contains~$\hat{e}$, and where for every angle $\angle uvw$
of $G$, at least one of the edges $(u,v)$ and $(v,w)$ is in $H$.
Subgraph $H$ can be found in linear time.
\end{lemma}
\begin{proof}
The technique to find $H$ is well known, see e.g.~\citep{BKL-CCCG96,BBDL01}, but we review it here to
prove the additional claims and verify that it works even if $G$ has parallel edges. 
Since $G$ is triangulated, its dual graph $G^*$ is 3-regular.   
Graph $G^*$ is also 2-connected, for otherwise
it (as a 3-regular graph) would have an edge $e^*$ whose removal disconnects $G^*$,
and this edge is dual to a loop of $G$.  By Petersen's theorem \citep{Pet1891},
$G^*$ therefore has a perfect matching $M^*$.   In fact (see e.g.~the proof of Petersen's theorem by \citet{BBDL01}) we can require
one edge of our choice to be in the perfect matching $M^*$;    we choose it here to be an edge that shares exactly one endpoint
with $\hat{e}^*$ so that $\hat{e}^*$ is \emph{not} in $M^*$.   This matching can be computed in linear time
\citep{BBDL01}.  

Define $M$ to be those edges of $G$ whose dual edges are in $M^*$, and let $H:=(V,E\setminus M)$, see also Figure~\ref{fig:triangulated}.
Since $M^*$ is a perfect matching, all faces of $H$ have degree 4, which
makes $H$ a bipartite subgraph of $G$.
Since the dual of $\hat{e}$ is not in $M^*$, reference-edge $\hat{e}$ is in $H$.
Consider any angle $\angle uvw$ of $G$, say incident to face $F$, and observe that its two edges $(u,v)$ and $(v,w)$ are distinct since
$G$ has no vertices of degree~1.  Since $M^*$ contains exactly one edge incident to $v_F$,
at most one of $(u,v)$ and $(v,w)$ can be in $M$, and the other is retained in $H$.
\end{proof}

\begin{figure}[ht]
\hspace*{\fill}
	\subfigure[]{\includegraphics[scale=1.1,page=1]{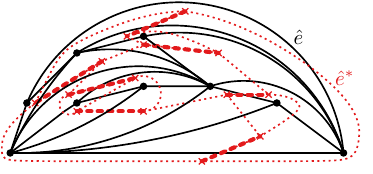}}
\hspace*{\fill}
	\subfigure[]{\includegraphics[scale=1.1,page=2]{why_new.pdf}}
\hspace*{\fill}
	\caption{\changed{(a) The triangulation from Figure~\ref{fig:intro}(b),}
	with its dual graph (red dotted) and a perfect matching in it (thick dashed).   
	\changed{(b) The} resulting bipartite graph $H$\changed{; its 2-coloring gives the face-hitting dominating sets from Figure~\ref{fig:intro}(a).}
	}
\label{fig:triangulated}
\end{figure}

\subsection{2-\changed{C}onnected Plane Graphs}
\label{sec:2connected}

Next we prove Theorem~\ref{thm:bipEdgeCover} for a 2-connected plane graph $G$
in the special case where $G$ has no bigons.
The plan is to find a \emph{plane triangulation} of $G$ (i.e., a supergraph $G^+$ that 
is triangulated and respects the embedding of $G$), 
and apply Lemma~\ref{lem:triangulated} to $G^+$.   For this
to work out, we need that some angles of $G$ are \emph{preserved} in $G^+$, i.e.,
no edge of $G^+\setminus G$ has been inserted into the angle, and so the angles are also
angles of $G^+$.    This concept can be applied to any plane supergraph $G^+$; when
$G^+$ is a triangulation then ``$\angle uvw$ has been preserved'' is equivalent
to what \citet{FIJR24} called being \emph{happy}:  $G^+$ contains a triangular face
that includes the edges $(u,v)$ and $(v,w)$ of the angle.
Call a vertex \emph{happy} if it has at least one incident happy angle.
The following result is crucial for proving Theorem~\ref{thm:bipEdgeCover},
but we defer its proof to Section~\ref{sec:makeHappy} where we will actually prove
a stronger result.

\begin{proposition}
\label{lem:makeHappy}
	Let $G$ be a 2-connected plane graph with \changed{at least three} vertices and no bigons.   Fix a vertex~$s$.   
Then $G$ has a triangulated plane supergraph $G^+$ without loops in which
all vertices except perhaps $s$ are happy.   Supergraph $G^+$ can be found in linear time.
\end{proposition}

We also need an easy observation about angles incident to faces, which \changed{implies
that there is no} need to explicitly ensure that faces have incident happy angles
(the approach taken by \citet{FIJR24}).

\begin{observation}
\label{obs:face_happy}
	Let $G$ be a 2-connected plane graph with \changed{at least three} vertices and no bigons, and let $G^+$ be any plane supergraph of $G$. 
Then every face $F$ of $G$ has at least two incident angles that are preserved in~$G^+$.
\end{observation}
\begin{proof}
The boundary $B$ of $F$ is a cycle with at least three vertices since $G$ is 2-connected without bigons.
Consider the subgraph $G_F$ of $G^+$ formed by $B$ and all edges that were inserted inside $F$.    This 
is a 2-connected graph with at least three vertices that can be drawn with all vertices on one face
(it is \emph{outerplanar}).   It is well-known that such a graph has at least two vertices of degree exactly~2.
Let $w$ be
one of these vertices.   It lies on $B$, hence its two incident edges in $G_F$ belong to $B$, and form an angle 
in $G^+$ since $G_F$ includes all edges that we added inside $F$.   So the angle incident to $F$ at $w$ is preserved.
\end{proof}

\begin{corollary}
\label{cor:2connected}
Theorem~\ref{thm:bipEdgeCover} holds for any 2-connected plane graph $G$ with $n\geq 2$ vertices and no bigons.
\end{corollary}
\begin{proof} 
Clearly the theorem holds if $n=2$ since $G$ itself is then bipartite, so we may assume that $n\geq 3$.   
Let $s$ be an endpoint of the reference-edge $\hat{e}$.
Using Proposition~\ref{lem:makeHappy}, find a triangulated plane supergraph $G^+$ of $G$ where all vertices except
perhaps $s$ are happy.   
Apply Lemma~\ref{lem:triangulated} to $G^+$ with respect to edge $\hat{e}$.   
This gives us, in linear time, a bipartite subgraph $H^+$ of $G^+$ that contains $\hat{e}$; in particular
vertex $s$ has $\hat{e}$ as an incident edge in $E(H^+)\cap E(G)$.   For every
vertex $v\neq s$ there exists an incident angle in $G^+$ 
where both edges belong to $G$.   At least one of these edges is retained
in $H^+$, so $v$ also has at least one incident edge in $E(H^+)\cap E(G)$.
Define $H$ to be the bipartite subgraph of $H^+$ with edge set $E(H^+)\cap E(G)$; by the
above $E(H)$ is an edge-cover. 
By Observation~\ref{obs:face_happy} edge set $E(H)$ is also face-hitting: For any face $F$ of $G$,
some incident angle was preserved in $G^+$,
at least one edge of this angle is retained in $H^+$, and hence $F$ has an incident edge in $E(H)$.
\end{proof} 

As it turns out, bigons do not pose a problem either.

\begin{corollary}
\label{cor:2connectedBigons}
	Theorem~\ref{thm:bipEdgeCover} holds for any 2-connected plane graph $G$ with \changed{at least two} vertices. 
\end{corollary}
\begin{proof}
Let $G'$ be the subgraph of $G$ obtained by deleting, for every bigon of $G$, one of the parallel edges that
bound the bigon (choosing one that is not reference-edge $\hat{e}$, if $\hat{e}$ is in a bigon).      Graph 
$G'$ is 2-connected (adding parallel edges cannot increase vertex-connectivity) and has no bigons, and so by 
Corollary~\ref{cor:2connected} we can find a bipartite subgraph $H'$ such that $E(H')$ is a face-hitting
edge cover of $G'$.   
Let $H$ be the bipartite subgraph of $G$ obtained by
adding to $H'$ all bigon-edges where one edge of the bigon was in $H'$.   
Clearly $H$ is bipartite and all steps to find it can be implemented in linear time.

We claim that $E(H)$ is $3^+$-face-hitting (and of course it is an edge cover since already $E(H')$ was one).   Let $F$ be
any $3^+$-face of $G$.   If $F$ was also a face of $G'$, then it was hit by $E(H')$ and so
also by $E(H)$.    If $F$ was not a face of $G'$, 
then it was replaced in $G'$ by a face $F'$ that is, roughly speaking, the union of $F$ and a number of bigons.
Formally, the facial boundary of $F'$ is the symmetric difference of the facial boundary of $F$ with
all those bigons where the deleted edge was incident to $F$.
Some edge $e'$ of $F'$ was in $E(H')$, hence $e'$ is either on $F$ or it belongs to a bigon $\{e',e''\}$
	with $e''$ incident to $F$.   In the latter case we added $e''$ to $H$. \changed{Therefore} either way \deleted{therefore} some
edge of $F$ is in $E(H)$.
\end{proof}

\subsection{Putting It All Together}
\label{sec:together}

We now put all ingredients together to prove Theorem~\ref{thm:bipEdgeCover} as follows.
%
We already know the result for 2-connected graphs (Corollary~\ref{cor:2connectedBigons}).
The way to generalize it to arbitrary connected
graphs is very standard:   Compute the bipartite subgraph for each 2-connected component of $G$,
and combine them to obtain a bipartite subgraph of $G$.
However, if $G$ is not simple then we have to be a bit more careful in
choosing the reference-edge so that a $3^+$-face is hit even if its boundary contains loops.

The details are as follows.
For ease of description choose the outer face of $G$ such that it contains the reference-edge $\hat{e}$.
Compute the 2-connected components of $G$; this can be done in linear time \citep{HT73}.
Consider a non-trivial 2-connected component $C$ (i.e., $C$ has at least two vertices).
If $C$ does not contain $\hat{e}$, then choose as reference-edge~$\hat{e}_C$ of $C$ an arbitrary edge on the outer face of $C$
in the induced planar embedding;
otherwise set $\hat{e}_C=\hat{e}$.  Apply 
Corollary~\ref{cor:2connectedBigons} to $C$ with $\hat{e}_C$ as reference-edge.
This gives a bipartite subgraph $H_C$ of $C$ for which $E(H_C)$ is a $3^+$-face hitting edge cover.
Let $H$ be the union of all these subgraphs; since 2-connected components intersect in at most one
	vertex and $H$ has no loops, this is bipartite.    Since (in a connected graph with \changed{at least two} vertices) every vertex belongs to 
at least one non-trivial 2-connected component $C$, and
is incident to an edge of $H_C$, the set $E(H)$ is an edge cover.

Now we must show that $E(H)$ hits any $3^+$-face $F$ of $G$.
	\changed{(Figure~\ref{fig:merge} illustrates some possibilities for $F$.)}
If $F$ is a face of one 2-connected component $C$, then $E(H_C)$ (and hence $E(H)$) hits $F$.
If $F$ is not a face of one 2-connected component, then let
$B$ be the boundary walk of $F$ and let $C_1,\dots,C_\ell$ for $\ell\geq 2$
be the 2-connected components intersected by $B$.      Note that $B\cap E(C_i)$ is a face of $C_i$
(in the induced planar embedding of $C_i$) for all $i$.   If this is a $3^+$-face of $C_i$ then
$E(H_{C_i})$ will include an edge of it and so $E(H)$ hits $F$.    If $B\cap E(C_i)$ is the outer
face of $C_i$, and $C_i$ is not a loop, then reference edge $\hat{e}_{C_i}$ was chosen on $B\cap E(C_i)$
and included in $E(H_{C_i})$, and so $E(H)$ hits $F$.   

\begin{figure}[ht]
\hspace*{\fill}
	{\includegraphics[scale=1,page=1]{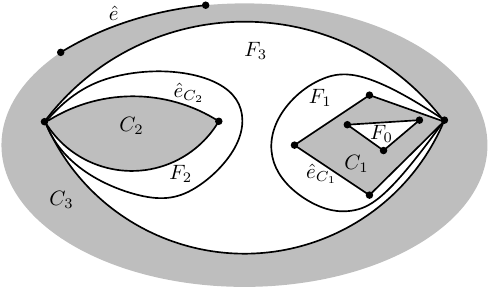}}
\hspace*{\fill}
	\changed{
	\caption{A graph that is not 2-connected; non-trivial 2-connected components are shaded.   Face $F_0$ is a face of $C_1$, hence is hit by $E(H_{C_1})$.   The boundary of $F_1$ meets $C_1$ in a $3^+$-face, hence $F_1$ is hit by $E(H_{C_1})$.     The boundary of face $F_2$ meets $C_2$ in the outer face of $C_2$, hence $F_2$ is hit by $\hat{e}_{C_2}$. Face $F_3$ is incident only to loops and a bigon of $C_3$ and not a $3^+$-face.}
	}
	\label{fig:merge}
\end{figure}

So we are done unless for $i=1,\dots,\ell$, either $B\cap E(C_i)$
is a loop, or $B\cap E(C_i)$ bounds a bigon of $C_i$ that is an inner face of $C_i$ and hence contains
face $F$ in its interior.   The latter can happen for at most one index $i$ since face $F$ cannot
be interior to two bigons that both share edges with it.   Therefore the boundary of $F$ consists
either only of loops, or of one bigon with some loops attached at its endpoints.
Either way $F$ has size at most 2 and is not a $3^+$-face.

\medskip
In summary, the edges of $H$ are a $3^+$-face-hitting edge cover, which proves Theorem~\ref{thm:bipEdgeCover}
and hence Theorem~\ref{thm:main} (subject to proving Proposition~\ref{lem:makeHappy}, which will happen in Section~\ref{subsec:allFaces}).

\subsection{Hitting faces of smaller sizes}

We briefly discuss here how the bounds change if we want to hit not only the $3^+$-faces but in fact \emph{all} faces 
with a dominating set $S$.   (This makes a difference only if the graph is allowed to be not simple.)

If we permit faces of degree 1 (i.e., faces whose boundary is a loop), then we must include the unique vertex of each such face
in $S$ to be face-hitting, and so $|S|=n$ may be required since we may have a degree-1 face at every vertex.   
See Figure~\ref{fig:bad}(a).

In fact, $|S|\approx n$ may be required even in a graph that does not have faces of degree 1 but may have loops.   A face
that is bounded by two parallel loops has degree 2, but still forces its unique vertex to be in any face-hitting vertex set $S$. 
We can build a graph that has such faces at $n-2$ vertices, without having a face of degree 1, see Figure~\ref{fig:bad}(b).
Hence $|S|\geq n-2$ may be required.

If we forbid loops, but permit bigons, then $|S|\geq \tfrac{3}{4}n$ may be required:   Take $n/4$
copies of $K_4$, replace every edge by a bigon, and then connect the copies with arbitrary further edges to make the
graph connected and planar (we can even achieve higher connectivity,  see Figure~\ref{fig:bad}(c)).
To hit the bigons at each $K_4$, set $S$ must include a \emph{vertex cover} of $K_4$, i.e.,
a vertex set $C$ such every edge has an endpoint in $C$.   Any vertex cover of $K_4$ must include
3 out of 4 vertices, and so $|S|\geq \tfrac{3}{4}n$.

\begin{figure}[ht]
\hspace*{\fill}
\subfigure[]{\includegraphics[scale=1.1,page=1]{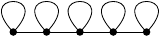}}
\hspace*{\fill}
\subfigure[]{\includegraphics[scale=1.1,page=2]{bad.pdf}}
\hspace*{\fill}
\subfigure[]{\includegraphics[scale=1.1,page=3]{bad.pdf}}
\hspace*{\fill}
\caption{Graphs where any face-hitting dominating set $S$ must be large.  (a) With faces of degree 1, we may need $|S|=n$.
(b) With loops, we may need $|S|\geq n-2$ even if all faces have degree 2 or more.
(c) With bigons, we may need $|S|\geq \tfrac{3}{4}n$ even if there are no loops.}
\label{fig:bad}
\end{figure}

On the positive side, for any connected plane graph $G$, any vertex cover is a face-hitting dominating set
since every vertex and face is incident to at least one edge.   If $G$ has no loops then it is 4-colourable,
and the three smallest colour classes form a vertex cover of size at most $\tfrac{3}{4}n$.   So
the lower bound for plane graphs without loops can be matched.
Note that due to the use of the 4-colour theorem, this proof only yields a quadratic-time (and difficult
to implement) algorithm to find the vertex cover. Reducing this run-time to linear would be equivalent to 
finding an independent set of size at least $n/4$ in any
planar graph in linear time, and no such results seem known (see \citet{BW05} for related discussion).

\section{Making Everyone Happy}
\label{sec:makeHappy}

To goal of this section is to prove Proposition~\ref{lem:makeHappy}, and more generally
to study the problem of finding a triangulated
plane supergraph of a given plane graph $G$ such that many (ideally all) vertices are happy.
As before, we assume that $G$ is connected and \changed{has at least two vertices}.
We also assume that all faces of $G$ have degree 3 or more, otherwise we cannot hope to make
the graph triangulated by adding edges.   

\changed{For ease of reference, we repeat and clarify here a few definitions. Let $G$ be a plane graph
and $G^+$ be a supergraph (not necessarily a triangulation).   We say that an angle $\angle uvw$ of $G^+$ is 
\emph{happy} if the incident face in $G^+$ is a triangle, and both edges $(u,v)$ and $(v,w)$ are edges of $G$.
We say that adding another edge $(u,w)$ to $G^+$ \emph{makes $\angle uvw$ happy} if it is happy afterwards due
to this edge, i.e., $(u,w)$ closed up a triangular face $\{u,v,w\}$ and both $(u,v)$ and $(v,w)$ are
edges of $G$.    
Note that once an angle $\angle uvw$ has been made happy, it will remain happy
even during later edge-additions since $\{u,v,w\}$ will remain a face. 
A vertex $v$ of $G^+$ is \emph{happy} there exists a happy angle $\angle uvw$,
and $v$ \emph{is made happy} by adding an edge $(u,w)$ if this made an angle $\angle uvw$ happy.}

\subsection{Some bad examples}

We first give some examples of graphs that justify the restrictions/permission made
on the input and output of Proposition~\ref{lem:makeHappy} (this will also be a convenient warm-up).
These graphs are illustrated in Figure~\ref{fig:moreBad}.
First, the proposition assumes that the input graph $G$
is 2-connected.   This is required: we can otherwise make only
a fraction of the vertices happy in some graphs.

\begin{lemma}
\label{lem:not1conn}
For $n$ divisible by 3, there exists a connected plane $n$-vertex graph $G$ 
where any plane triangulation has at most $\tfrac{2}{3}n$ happy vertices.
\end{lemma}
\begin{proof}
Let $C$ be a cycle with $n/3$ vertices, and attach 
two vertices $x,y$ of degree~1
at every vertex $v$ of the cycle.   Fix a plane embedding and
outerface such that $x$ is inside $C$ and $y$ is outside.   Now consider
	\changed{adding edges to create some}
plane triangulation of this graph.   To make vertex $x$ happy,
we must make its unique incident angle happy, which can be done only by
adding a loop at $v$ that encloses $x$.   (So in particular, we cannot make
$x$ happy at all if adding loops is forbidden.)   To make vertex $y$ happy,
we must add another loop that encloses $y$, and the two loops together add
edges inside all four angles incident to $v$.   So if $x,y$ are both happy
then $v$ cannot be happy, meaning that out of each triple of vertices at
	least one is \changed{never made} happy. 
\end{proof}

Next, the proposition permits parallel edges in the triangulation.
Again this is required: we can otherwise make
only a constant number of vertices happy in some graphs.

\begin{lemma}
\label{lem:notSimple}
For any $n\geq 4$, there exists a 2-connected plane $n$-vertex graph $G$ 
where any simple triangulation has at most four happy vertices.
\end{lemma}
\begin{proof}
Consider the graph $K_{2,n-2}$, i.e., two vertices $a,a'$ are both adjacent
	to all of $n{-}2$ vertices $b_0,\dots,b_{n-3}$.   This has a unique \changed{rotation scheme} (up to renaming), and
its faces
are  bounded by the 4-cycles $\{a,b_i,a',b_{i+1}\}$ for $i=0,\dots,n{-}3$, addition modulo $n{-}2$.   
To triangulate such a face, we can either insert $(a,a')$ or $(b_i,b_{i+1})$, and a vertex $b_i$
	is \changed{made} happy only if it is incident to a face where we inserted $(a,a')$.   But in a 
	simple triangulation  we can only have one copy of $(a,a')$, so at most \changed{$a,a'$ and} two
of the vertices $b_0,\dots,b_{n-3}$ can be happy.
\end{proof}

Finally, the proposition permits that one vertex $s$ is not happy.  Again this
is required.

\begin{lemma}
\label{lem:notOddCycle}
For any $n\geq 5$ odd, any triangulation of the $n$-cycle 
has at most $n-1$ happy vertices.
\end{lemma}
\begin{proof}
If all vertices were happy, then in one of the two faces we
would have at least $\lceil n/2 \rceil$ happy angles, and by $n$ odd
and the pigeonhole principle
hence have two consecutive ones.   This is impossible, since the edge that
makes an angle happy necessarily gets inserted into the two adjacent angles
of the face.
\end{proof}

\begin{figure}[ht]
\hspace*{\fill}
\subfigure[]{\includegraphics[width=0.32\linewidth,page=4]{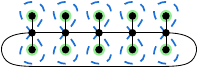}}
\hspace*{\fill}
\subfigure[]{\includegraphics[width=0.32\linewidth,page=5]{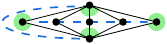}}
\hspace*{\fill}
\subfigure[]{\includegraphics[width=0.32\linewidth,page=6]{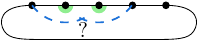}}
\hspace*{\fill}
\caption{Graphs where triangulations have few happy angles (shaded green).
(a) A graph with cutvertices.   (b) Demanding a simple triangulation. (c)
We cannot make two consecutive angles happy.}
\label{fig:moreBad}
\end{figure}

Surprisingly, as we will prove below, $C_n$ (for $n\geq 5$ odd) is the \emph{only} 2-connected
graph where we have to have an unhappy vertex if the triangulation is allowed to have
parallel edges.

\subsection{Triangulating a single face}

Towards the proof of Proposition~\ref{lem:makeHappy}, 
we first study how to triangulate a single face (modelled here as
triangulating the interior of a cycle).   We show a stronger
statement: we can 
\deleted{exactly}%
prescribe which angles become happy
(presuming none of them are consecutive, which is obviously necessary).
\changed{In fact, we can even prescribe which angles should \emph{not}
be happy (as long as at least two are allowed to be happy); this is
not a property that will be used later, but we find it interesting in
its own right.}

\begin{lemma}
\label{lem:one_face}
Let $C_n$ be an $n$-cycle for $n\geq 3$, and
let $A$ be any subset of the vertices such that no two of them 
are consecutive on the cycle.   
There exists a simple triangulation of the interior of $C_n$
such that all interior angles at vertices of $A$ are happy.
It can be found in linear time.

If $|A|\geq 2$ then we can furthermore ensure that 
an interior angle of $C_n$ is happy if and
	only if it \deleted{is} belongs to a vertex of $A$.
\end{lemma}
\begin{proof}
Nothing is to show if $n=3$ since then $C_n$ is already a triangulated graph
and all angles are happy.   So assume that $n\geq 4$ and
pad $A$, if needed, with extra vertices to achieve $|A|\geq 2$.
If $n=4$ then by $|A|\geq 2$ and assumption we alternatingly have
vertices in $A$ and not in $A$ along $C_n$; adding an edge between
the two vertices not in $A$ then satisfies all conditions.

	So assume that $n\geq 5$ \changed{and consult Figure~\ref{fig:cycle} for an illustration of the following steps.}
	We first add, for every vertex $v\in A$, 
an edge inside $C_n$ between its two neighbours $v',v''$ on $C_n$; the naming
is such that $v'$ is the ccw neighbour of $v$.
This makes the interior angle incident to $v$ happy and does not add a bigon by $n\geq 5$.  
It also maintains planarity since vertices of $A$ are not consecutive along $C_n$,
and so none of the added edges cross each other.   

Let $F$ be the face of the resulting graph that is incident to all
added edges.  If we did not care about
	accidentally \changed{creating further happy vertices}, then we could simply triangulate $F$ arbitrarily.   
To make \emph{only} the vertices of $A$ happy, we have to be a bit
more careful.  Recall that $|A|\geq 2$, and let $v,w$ be two vertices in $A$.
Their neighbours $v',v'',w',w''$ are incident to $F$ (where possibly $v''=w'$ or
$w''=v'$).    If $w'\neq v''$ then add an edge $(v',w')$ (if $w'=v''$ then this
edge exists already).  Then add edges from $v'$ to everyone
that is strictly between $v''$ and $w'$ on $F$ in clockwise order, and add edges from $w'$
to everyone that is strictly between $w''$ and $v'$ on $F$ in clockwise order.
This does not make any other interior angle of $C_n$
happy since each of the resulting triangular faces has at most one edge on $C_n$.        
\end{proof}

\begin{figure}[ht]
\hspace*{\fill}
\subfigure[]{\includegraphics[scale=1,page=1]{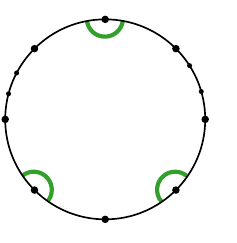}}
\hspace*{\fill}
\subfigure[]{\includegraphics[scale=1,page=2]{ear_journal.pdf}}
\hspace*{\fill}
\subfigure[]{\includegraphics[scale=1,page=3]{ear_journal.pdf}}
\hspace*{\fill}
	\caption{\changed{Triangulating the interior such that a} prescribed set of angles (indicated by green circles) \changed{is} happy.}
\label{fig:cycle}
\end{figure}

\subsection{Triangulating all faces}
\label{subsec:allFaces}

We know from Lemma~\ref{lem:notOddCycle} that odd-length cycles will be a
special case for Proposition~\ref{lem:makeHappy}, and so we prove the result
for them separately.

\begin{observation}
\label{cor:odd_cycle}
Proposition~\ref{lem:makeHappy} holds for any odd-length cycle $C_n$. 
\end{observation}
\begin{proof}
Recall that we are given one vertex $s$ that need not be made happy.
Enumerate the vertices of $C_n$ as 
	\changed{$s,v_1,\dots,v_{n-1}$.}
	Apply Lemma~\ref{lem:one_face} to triangulate the interior of $C_n$ to make 
	\changed{$v_2,v_4,\dots,v_{n-1}$} happy, and apply it again to triangulate the exterior of
	$C_n$ to make \changed{$v_1,v_3,\dots,v_{n-2}$} happy.   This can
clearly be implemented in linear time.
\end{proof}

As promised, for all other 2-connected plane graphs we can make all vertices
happy.  Put differently, we now prove Theorem~\ref{thm:make_happy},
which together with Observation~\ref{cor:odd_cycle} implies Proposition~\ref{lem:makeHappy}.    
We restate the theorem here for ease of reference, and then devote the rest of the 
subsection to proving it.

\ThmMakeHappy*

Crucial to the proof of this theorem is the idea of building up the graph $G$
``piece by piece'', adding one face at a time, and triangulating this face 
while maintaining a suitable invariant. 
This ``building up the graph'' uses a well-known tool that we define first.

\begin{definition}
An \emph{ear decomposition} of a graph $G$ is a collection $\langle P_1,\dots,P_f \rangle$ of paths (``\emph{ears}'') in $G$ with the following properties:
\begin{itemize}
\item Every edge of $G$ belongs to exactly one path.
\item $P_1$ is a cycle, i.e., its endpoints coincide.
\item For $i>1$, $P_i$ is a path with two distinct endpoints.   
	The endpoints of $P_i$ are in $V_{i-1}{:=}V(P_1){\cup}\dots \allowbreak{\cup}V(P_{i-1})$, while
	the interior vertices of $P_i$ (if any) are not in $V_{i-1}$.
\end{itemize}
\end{definition}

It is well-known that a graph $G$ with \changed{at least three} vertices
is 2-connected if and only if it has an ear decomposition \citep{Whitney31b}.    
Given an ear decomposition, we define (for $i\geq 1$) the \emph{partial graph $G_i$} to be the graph with vertices $V_i$ 
and all edges of $P_1,\dots,P_i$; note that $G_i$ is 2-connected since it has an ear decomposition.   

It is known that if $G$ is a plane graph, then we can put extra conditions on
ear decompositions.   We have not been able to find an explicit reference for this result,
and hence provide a short proof, mostly based on the insights in \citet{RT86}
and illustrated in Figure~\ref{fig:ear_planar}.

\begin{lemma}
\label{lem:earDecompPlanar}
	Let $G$ be a 2-connected plane graph with \changed{at least three} vertices.   Then for any face $F$ of $G$, 
there exists an ear decomposition $\langle P_1,\dots,P_f \rangle$ of $G$ 
and a choice of outer face of $G$ such that
$P_1$ bounds face $F$, and for $i>1$ path $P_i$ lies on
	the outer face of \changed{the partial graph} $G_i$ in the induced embedding of $G_i$.
Such an ear decomposition can be found in linear time.
\end{lemma}
\begin{proof}
Let $(s,t)$ be an arbitrary edge on $F$.   Since $G$ is 2-connected, it has
a vertex ordering $v_1,\dots,v_n$ such that $v_1=s$, $v_n=t$, and every
other vertex has at least one earlier and at least one later neighbour \citep{LEC67}.
This can be computed in linear time \citep{ET76}.   
Direct every edge from the lower-indexed to the higher-indexed endpoint to obtain 
a \emph{bipolar orientation} \citep{RT86}, i.e., 
an acyclic edge-orientation where $s$ is the only source and $t$
is the only sink.   Consider the dual graph $G^*$, and direct it such that every dual edge $e^*$
crosses the corresponding edge $e$ of $G$ left-to-right.   
If we delete the edge dual to $(s,t)$ from $G^*$, then the resulting orientation is also a bipolar
orientation, with the source and sink the two faces incident to $(s,t)$ \citep{RT86}.   
Up to reversal of edges we may
therefore assume that $v_F$ is the source of $G^*$ while the vertex of the outer face is the sink.
Compute a topological order of $G^*$ to enumerate the faces of $G$ as $F_1,\dots,F_{f+1}$, with
$F_1=F$ and $F_{f+1}$ the outer face.   

Set $P_1$ to be the boundary of $F{=}F_1$.  To define $P_i$ for $2\leq i\leq f$,
consider the interior face $F_i$. Since we have a bipolar orientation of $G^*$,
vertex $v_{F_i}$ is neither a source nor a sink.  It is known that
its outgoing edges are consecutive in the clockwise order of edges, as are its incoming edges
\citep{TT86}.
Translated from $G^*$ to $G$, therefore the boundary of $F_i$ is split 
into two non-empty parts, one part with edges shared with smaller-indexed faces (hence in earlier paths), and one part 
with edges shared with larger-indexed faces.  Define $P_i$ to be to be this latter part.  One easily verifies all required
properties.
\end{proof}

\begin{figure}[ht] \centering
\hspace*{\fill}
\includegraphics[scale=0.5,page=2,trim=0 0 0 0,clip]{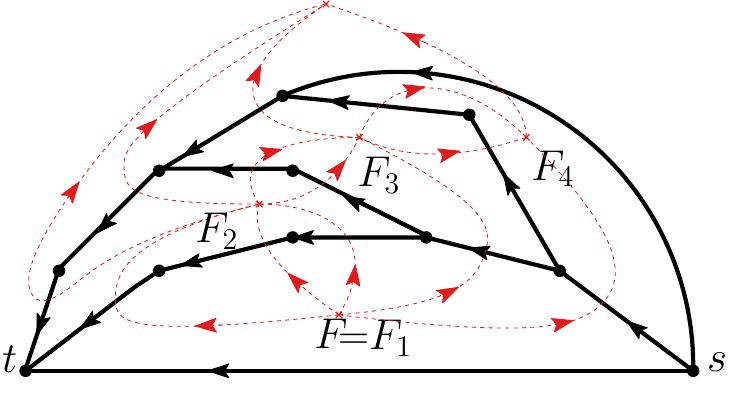}
\hspace*{\fill}
\includegraphics[scale=0.5,page=3,trim=0 0 0 0,clip]{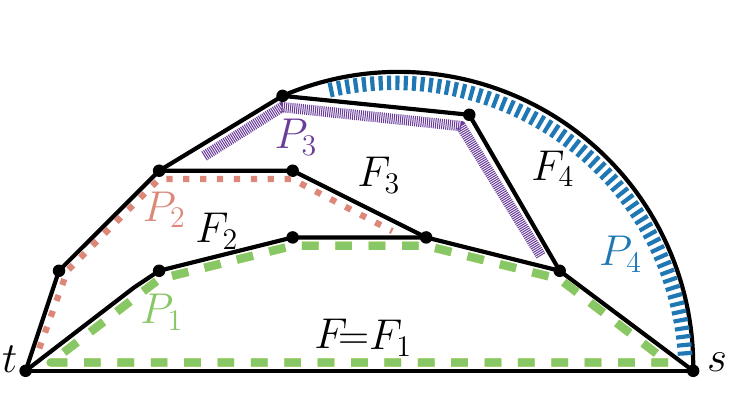}
\hspace*{\fill}
\caption{Finding an ear decomposition of a plane 2-connected graph via a topological order in the bipolar orientation
	of the dual graph.}
\label{fig:ear_planar}
\hspace*{\fill}
\end{figure}


Continuing the proof of Theorem~\ref{thm:make_happy}, we compute an ear decomposition 
$\langle P_1,\dots,P_f \rangle$ 
of the given graph $G$ by using Lemma~\ref{lem:earDecompPlanar}, choosing as face $F$ 
one of even degree if there are such faces, and an arbitrary one otherwise.
This also enumerates the interior faces as $F_1,\dots,F_f$, with $P_i$ on the boundary
of $F_i$ for $i=1,\dots,f$.
As outlined earlier, the plan is to triangulate $F_i$, for $i=1,\dots,f$,
such that the resulting supergraph $G_i^+$ of $G_i$ satisfies the following
{\bf invariant}:
\begin{enumerate}[(a)]
\item Every interior vertex of $G_i^+$ is happy. \label{it:interior}
\item For any edge $e$ on the outer face of $G_i^+$, at least one endpoint of $e$ is \emph{interior-happy},
	i.e., it is incident to a happy angle at an interior face.
	\label{it:outer}
\end{enumerate}
Since we build $G_i^+$ by triangulating $F_1,\dots,F_i$, and always do so
by appealing to Lemma~\ref{lem:one_face}, it
also follows automatically that the outer face of $G_i^+$ is the same as the
outer face of $G_i$, every interior face of $G_i^+$ has degree 3, and no loops are added.

\medskip\noindent{\bf Base case if $G$ has a face of even degree:}  
If $G$ has a face of even degree, then with our choice of ear decomposition face
$F_1$ has even degree.   See also Figure~\ref{fig:baseCase}(a).
We can then appeal to Lemma~\ref{lem:one_face} to triangulate $F_1$ while
making every second interior angle happy.    Since $G_1$ has no interior vertices, the
invariant then clearly holds for the resulting $G_1^+$.

\medskip\noindent{\bf Base case if $G$ has only faces of odd degree:}  
In this case we will triangulate faces $F_1$ and $F_2$ at the same time such
that the invariant holds for $G_2^+$.   Note first that $F_2$ must
exist since otherwise $G$ would be a cycle of odd length.   It also must share some
edges with $F_1$ since the endpoints of $P_2$ are distinct. 
So $G_2$ consists of three disjoint paths, call them $\pi_a=\langle a_1,\dots,a_p \rangle$,
$\pi_b=\langle b_1,\dots,b_q \rangle$, and
$\pi_c=\langle c_1,\dots,c_r \rangle$, with $a_1=b_1=c_1$ and $a_p=b_q=c_r$.
We assume that the naming is such that $F_1$ borders $\pi_a$ and $\pi_b$
while $F_2$ borders $\pi_b$ and $\pi_c$, hence the interior vertices of $\pi_b$
are not on the outer face of $G_2$.   See also Figure~\ref{fig:baseCase}(b) and (c).
We know that $p+q$ is odd, as is $q+r$, since all faces have odd degree.
In consequence the outer face of $G_2$ (which is bounded by the cycle
formed by $\pi_a\cup \pi_c$) has even degree.   Set $A_0$
to be every other vertex on the outer face, choosing among the two possibilities
such that $a_1$ is \emph{not} in $A_0$.   
For $i\in\{1,2\}$, write $A_0^i$ for all those vertices  in $A_0$ that are incident to face $F_i$.

We now have two subcases.   If $q$ is even (hence $p$ and $r$ are odd) then
vertex $b_q$ is \emph{not} in $A_0$.   Set $A_1$ to be $\{b_3,b_5,\dots,b_{q-1}\}$    
and $A_2$ to be $\{b_2,b_4,\dots,b_{q-2}\}$.
Observe that for $i\in \{1,2\}$ the
set $A_0^i\cup A_i$ consists of vertices on $F_i$, no two of which are consecutive.
We can hence appeal to Lemma~\ref{lem:one_face} to make the angles incident to $F_i$ at $(A_0\cap V(F_i))\cup A_i$ 
happy by triangulating $F_i$.   Thus in the resulting graph $G_2^+$ all interior vertices
of $\pi_b$ are happy, and every other vertex on the outer face is interior-happy, and the
invariant holds.

If $q$ is odd, then define $A_1$ and $A_2$ similarly, but this time $A_1:=\{b_3,b_5,\dots,b_q\}$
includes $b_q$ while
$A_2=\{b_2,b_4,\dots,b_{q-1}\}$.   Observe that we had $b_q\in A_0$, and we now, for $i=1,2$,
appeal to Lemma~\ref{lem:one_face} to triangulate face $F_i$ using set $(A_0^i\setminus b_q)\cup A_i$ after verifying
that no two of these vertices are consecutive on $F_i$.   This makes all interior vertices of
$\pi_b$ happy in $G_3^+$.   It also makes all vertices of $A_0$ happy since $b_q\in A_1$, and
the invariant holds.

\begin{figure}[ht]
\hspace*{\fill}
\subfigure[]{\includegraphics[scale=1,page=4]{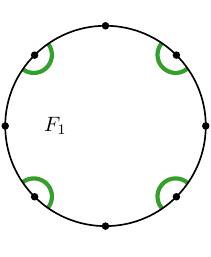}}
\hspace*{\fill}
\subfigure[]{\includegraphics[scale=1,page=5]{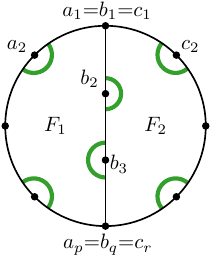}}
\hspace*{\fill}
\subfigure[]{\includegraphics[scale=1,page=6]{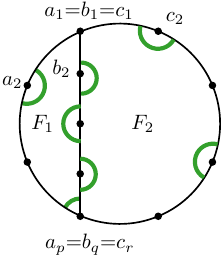}}
\hspace*{\fill}
\caption{Three base cases:  (a) $F_1$ has even degree.   (b) All faces have
odd degree, and $q$ is even.   (c) All faces have odd degree, and $q$ is odd.  }
\label{fig:baseCase}
\end{figure}

\medskip\noindent{\bf Step:}   Assume now that $G_{i-1}^+$ (for some $i\geq 2$) satisfies the invariant.
Let $P_i=\langle y_0,y_1,\dots,y_k,y_{k+1}\rangle$ for some $k\geq 0$ be the next ear, so
$y_0\neq y_{k+1}$ are vertices of $G_{i-1}$ while
all other vertices of $P_i$ (if any) are new.   Enumerate the interior face $F_i$ incident to $P_i$ as 
$$F_i=\langle y_0,y_1,\dots,y_k,y_{k+1}{=}z_0,z_1,\dots,z_\ell,z_{\ell+1}=y_0\rangle,$$
so $\langle z_0,z_1,\dots,z_\ell,z_{\ell+1}\rangle$ for some $\ell\geq 0$ was a path along the outer face of $G_{i-1}$ between the ends of $P_i$.
The plan is again to appeal to Lemma~\ref{lem:one_face}, and to this end, we define
a set $A$ as follows (see also Figure~\ref{fig:step}).

First we want to ensure Invariant (\ref{it:interior}), i.e., that all interior vertices of $G_i^+$ are happy.   
	This holds for all such vertices except $z_1,\dots,z_\ell$ by induction.   So set $A_1$ to
	be the set of all vertices of $z_1,\dots,z_\ell$ that were not interior-happy in $G_{i-1}^+$.
	Since the invariant held for $G_{i-1}^+$, no two of them were consecutive on the outer face of $G_{i-1}$,
	hence no two of them are consecutive on $F_i$, not even if $k=0$ since $y_0,y_{k+1}\not\in A_1$.

Next we want to ensure Invariant (\ref{it:outer}), i.e., that for any edge on the outer face at least one endpoint is interior-happy.
	This holds for all edges except the ones of $P_i$ by induction.     To ensure it for $P_i$, we have
	have two cases.   
\begin{enumerate}
\item 

	In the first case, one of $y_0$ and $y_{k+1}$ (say $y_{k+1}$) is interior-happy already.   
	We set $A_2:=\{y_1,y_3,\dots,\allowbreak y_{2\lceil k/2 \rceil-1}\}$, i.e., we take
	every odd-indexed vertex of $P_i$ until $y_{k-1}$ or $y_k$, depending on the parity of $k$.
	See Figure~\ref{fig:step}(a-b).
	Note that every edge in $P_i$ has one endpoint in $A_2\cup \{y_{k+1}\}$, 
	so creating happy angles at the vertices of $A_2$ suffices to satisfy (\ref{it:outer})
	since $y_{k+1}$ is already interior-happy.
	Also none of the vertices in $A_2$ is consecutive on $F_i$ with a vertex of $A_1$ 
	since $y_0,y_{k+1}\not\in A_1$.

\item
	In the second case, neither $y_0$ nor $y_{k+1}$ was already interior-happy.   This implies that
	$\ell\geq 1$, since by Invariant (b) they then must not be consecutive on the outer face of $G_{i-1}^+$.
	Again using Invariant~(b), this implies that $z_1$ and $z_\ell$ \emph{were} interior-happy,
	and in particular, do not belong to $A_1$.
	Set $A_2:=\{y_0,y_2,\dots,y_{2\lceil k/2 \rceil}\}$, i.e., we take every even-indexed
	vertex of $P_i$ until $y_k$ or $y_{k+1}$, depending on the parity of $k$.   
	See Figure~\ref{fig:step}(c).
	Note that none of these vertices
	(not even $y_0$ or $y_{k+1}$) are consecutive in $F_i$ with a vertex of $A_1$ since $z_1,z_\ell\not\in A_1$.
\end{enumerate}

So we have chosen $A_1$ and $A_2$ such that no two vertices of $A:=A_1\cup A_2$ are consecutive
	along $F_i$, and such that creating happy angles at the vertices in $A$ ensures the invariant. 
	We hence use Lemma~\ref{lem:one_face} with $A$ to triangulate $F_i$ and obtain the desired graph $G_i^+$ 

\begin{figure}[ht]
\hspace*{\fill}
\subfigure[]{\includegraphics[width=0.25\linewidth,page=7]{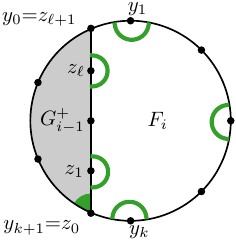}}
\hspace*{\fill}
\subfigure[]{\includegraphics[width=0.25\linewidth,page=8]{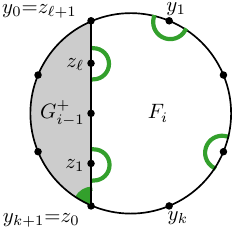}}
\hspace*{\fill}
\subfigure[]{\includegraphics[width=0.25\linewidth,page=9]{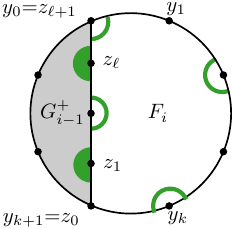}}
\hspace*{\fill}
\caption{Three examples of the step:  (a) $y_{k+1}$ was interior-happy, $k$ odd. (b) $y_{k+1}$ was interior-happy, $k$ even.
(c) Neither end of $P_i$ was interior-happy.  Vertices that were already happy in $G_{i-1}^+$ are indicated with
filled circular arcs.}
\label{fig:step}
\end{figure}

\medskip\noindent{\bf Triangulating the outer face.}
Using the algorithm inherent in the above inductive proof, we
can build $G_n^+$ in linear time since triangulating face $F_i$ takes time proportional
to its degree.    With this all interior vertices are happy.
Let $A$ be the set of vertices on the outer face of $G_n^+$ that are not interior-happy.   By Invariant~(b)
no two of these are consecutive on the outer face, so by appealing to Lemma~\ref{lem:one_face} once more with the outer face
and $A$, we have hence proved Theorem~\ref{thm:makeHappy} and Proposition~\ref{lem:makeHappy}.

\section{Further Thoughts}
\label{sec:odds}

In this paper, we gave a new proof that every simple planar connected graph with \changed{at least two} vertices has a vertex partition into two
face-hitting dominating sets, with the goal of developing a linear-time algorithm to find this vertex partition.  
The proof extends to graphs that are not simple, as long as we are content with hitting only those faces that 
are incident to at least three distinct vertices.    We also studied one crucial ingredient, the problem of
triangulating a given plane graph such that (almost) all vertices are happy.   We characterized exactly when
this can be done for a 2-connected graph if parallel edges are allowed, and gave various examples where it
cannot be done if the graph is not 2-connected or the triangulation must be simple.

One natural open problem concerns bounds for face-hitting dominating sets in 4-connected planar graphs.
The lower-bound example of  \citet{BSTZ97} has 3-cycles that are not faces, i.e., it is not 4-connected.   Can we improve the
upper bound on the size of face-hitting dominating sets for 4-connected planar graphs, at least for triangulations?
Another open problem concerns making (almost) all vertices happy in plane graphs that are not necessarily 2-connected.
Are vertices of degree 1 the only obstacle, so can we make (almost) all vertices happy in a connected plane graph
with minimum degree 2, or at least  in any 2-edge-connected plane graph?


\bibliographystyle{plainnat}
\bibliography{refs}

@article{Baker94,
  author       = {Brenda S. Baker},
  title        = {Approximation Algorithms for {NP}-Complete Problems on Planar Graphs},
  journal      = {J. {ACM}},
  volume       = {41},
  number       = {1},
  pages        = {153--180},
  year         = {1994},
  doi          = {10.1145/174644.174650},
  bibsource    = {dblp computer science bibliography, https://dblp.org}
}

@article{BEEMT93,
  author       = {Marshall W. Bern and
                  Herbert Edelsbrunner and
                  David Eppstein and
                  Scott A. Mitchell and
                  Tiow Seng Tan},
  title        = {Edge Insertion for Optimal Triangulations},
  journal      = {Discret. Comput. Geom.},
  volume       = {10},
  pages        = {47--65},
  year         = {1993},
  doi          = {10.1007/BF02573962},
  bibsource    = {dblp computer science bibliography, https://dblp.org}
}

@article{BW05,
  author       = {Therese Biedl and
                  Dana F. Wilkinson},
  title        = {Bounded-Degree Independent Sets in Planar Graphs},
  journal      = {Theory Comput. Syst.},
  volume       = {38},
  number       = {3},
  pages        = {253--278},
  year         = {2005},
  doi          = {10.1007/S00224-005-1139-0},
  bibsource    = {dblp computer science bibliography, https://dblp.org}
}

@article{BBDL01,
  author       = {Therese Biedl and
                  Prosenjit Bose and
                  Erik D. Demaine and
                  Anna Lubiw},
  title        = {Efficient Algorithms for {P}etersen's Matching Theorem},
  journal      = {J. Algorithms},
  volume       = {38},
  number       = {1},
  pages        = {110--134},
  year         = {2001},
  doi          = {10.1006/JAGM.2000.1132},
  bibsource    = {dblp computer science bibliography, https://dblp.org}
}

@article{BKK97,
  author       = {Therese Biedl and
                  Goos Kant and
                  Michael Kaufmann},
  title        = {On Triangulating Planar Graphs Under the Four-Connectivity Constraint},
  journal      = {Algorithmica},
  volume       = {19},
  number       = {4},
  pages        = {427--446},
  year         = {1997},
  doi          = {10.1007/PL00009182},
  bibsource    = {dblp computer science bibliography, https://dblp.org}
}

@inproceedings{BKL-CCCG96,
  author       = {Prosenjit Bose and
                  David G. Kirkpatrick and
                  Zaiqing Li},
  editor       = {Frank Fiala and
                  Evangelos Kranakis and
                  J{\"{o}}rg{-}R{\"{u}}diger Sack},
  title        = {Efficient Algorithms for Guarding or Illuminating the Surface of a
                  Polyhedral Terrain},
  booktitle    = {Canadian Conference on Computational Geometry},
  pages        = {217--222},
  publisher    = {Carleton University Press},
  year         = {1996},
  url          = {http://www.cccg.ca/proceedings/1996/cccg1996\_0037.pdf},
  bibsource    = {dblp computer science bibliography, https://dblp.org}
}

@article{BSTZ97,
  author       = {Prosenjit Bose and
                  Thomas C. Shermer and
                  Godfried T. Toussaint and
                  Binhai Zhu},
  title        = {Guarding Polyhedral Terrains},
  journal      = {Comput. Geom.},
  volume       = {7},
  pages        = {173--185},
  year         = {1997},
  doi          = {10.1016/0925-7721(95)00034-8},
  bibsource    = {dblp computer science bibliography, https://dblp.org}
}

@inproceedings{CRR24,
  author       = {Aleksander B. G. Christiansen and
                  Eva Rotenberg and
                  Daniel Rutschmann},
  editor       = {David P. Woodruff},
  title        = {Triangulations Admit Dominating Sets of Size 2\emph{n}/7},
  booktitle    = {{ACM-SIAM} Symposium on Discrete Algorithms,
                  {SODA} 2024},
  pages        = {1194--1240},
  publisher    = {{SIAM}},
  year         = {2024},
  doi          = {10.1137/1.9781611977912.47},
  bibsource    = {dblp computer science bibliography, https://dblp.org}
}

@book{Die12,
  author       = {Reinhard Diestel},
  title        = {Graph Theory, 4th Edition},
  series       = {Graduate texts in mathematics},
  volume       = {173},
  publisher    = {Springer},
  year         = {2012},
  isbn         = {978-3-642-14278-9},
  bibsource    = {dblp computer science bibliography, https://dblp.org}
}

@article{EGSW21,
  author       = {William S. Evans and
                  Ellen Gethner and
                  Jack Spalding{-}Jamieson and
                  Alexander Wolff},
  title        = {Angle Covers: Algorithms and Complexity},
  journal      = {J. Graph Algorithms Appl.},
  volume       = {25},
  number       = {2},
  pages        = {643--661},
  year         = {2021},
  doi          = {10.7155/JGAA.00576},
  bibsource    = {dblp computer science bibliography, https://dblp.org}
}

@article{ET76,
  author       = {Shimon Even and
                  Robert Endre Tarjan},
  title        = {Computing an \emph{st}-Numbering},
  journal      = {Theor. Comput. Sci.},
  volume       = {2},
  number       = {3},
  pages        = {339--344},
  year         = {1976},
  doi          = {10.1016/0304-3975(76)90086-4},
  bibsource    = {dblp computer science bibliography, https://dblp.org}
}

@inproceedings{FIJR24,
  author       = {P. Francis and
                  Abraham M. Illickan and
                  Lijo M. Jose and
                  Deepak Rajendraprasad},
  editor       = {Daniel Kr{\'{a}}l and
                  Martin Milanic},
  title        = {Face-Hitting Dominating Sets in Planar Graphs},
  booktitle    = {International Workshop on Graph-Theoretic Concepts in Computer Science,
                  {WG} 2024},
  series       = {Lecture Notes in Computer Science},
  pages        = {211--219},
  publisher    = {Springer},
  year         = {2024},
  doi          = {10.1007/978-3-031-75409-8\_15},
  bibsource    = {dblp computer science bibliography, https://dblp.org}
}

@book{GJ79,
  author       = {Michael R. Garey and
                  David S. Johnson},
  title        = {Computers and Intractability: {A} Guide to the Theory of NP-Completeness},
  publisher    = {W. H. Freeman},
  year         = {1979},
  isbn         = {0-7167-1044-7},
  bibsource    = {dblp computer science bibliography, https://dblp.org}
}

@article{HT73,
  author       = {John E. Hopcroft and
                  Robert Endre Tarjan},
  title        = {Efficient Algorithms for Graph Manipulation {[H]} (Algorithm 447)},
  journal      = {Commun. {ACM}},
  volume       = {16},
  number       = {6},
  pages        = {372--378},
  year         = {1973},
  doi          = {10.1145/362248.362272},
  bibsource    = {dblp computer science bibliography, https://dblp.org}
}

@article{KB97,
  author       = {Goos Kant and
                  Hans L. Bodlaender},
  title        = {Triangulating Planar Graphs while Minimizing the Maximum Degree},
  journal      = {Inf. Comput.},
  volume       = {135},
  number       = {1},
  pages        = {1--14},
  year         = {1997},
  doi          = {10.1006/INCO.1997.2635},
  bibsource    = {dblp computer science bibliography, https://dblp.org}
}

@inProceedings{LEC67,
	author =	"Abraham Lempel and Shimon Even and Israel Cederbaum",
	title =		"An algorithm for planarity testing of graphs",
	booktitle =	"Theory of Graphs, International Symposium Rome 1966",
	editor =	"Pierre Rosenstiehl",
	publisher =	"Gordon and Breach",
	year =		"1967",
	pages =		"215-232"
}

@article{MT96,
  author       = {Lesley R. Matheson and
                  Robert Endre Tarjan},
  title        = {Dominating Sets in Planar Graphs},
  journal      = {Eur. J. Comb.},
  volume       = {17},
  number       = {6},
  pages        = {565--568},
  year         = {1996},
  doi          = {10.1006/EUJC.1996.0048},
  bibsource    = {dblp computer science bibliography, https://dblp.org}
}

@article{Pet1891,
        author =        "Julius Petersen",
        title =         "{Die Theorie der regul\"aren graphs (The theory of
                                regular graphs).}",
        journal =       "Acta Mathematica",
        volume =        "15",
        pages =         "193-220",
        year =          "1891"
}

@article{RSST97,
  author       = {Neil Robertson and
                  Daniel P. Sanders and
                  Paul D. Seymour and
                  Robin Thomas},
  title        = {The Four-Colour Theorem},
  journal      = {J. Comb. Theory {B}},
  volume       = {70},
  number       = {1},
  pages        = {2--44},
  year         = {1997},
  doi          = {10.1006/JCTB.1997.1750},
  bibsource    = {dblp computer science bibliography, https://dblp.org}
}

@article{RT86,
  author       = {Pierre Rosenstiehl and
                  Robert Endre Tarjan},
  title        = {Rectilinear Planar Layouts and Bipolar Orientations of Planar Graphs},
  journal      = {Discret. Comput. Geom.},
  volume       = {1},
  pages        = {343--353},
  year         = {1986},
  doi          = {10.1007/BF02187706},
  bibsource    = {dblp computer science bibliography, https://dblp.org}
}

@article{TT86,
  author       = {Roberto Tamassia and
                  Ioannis G. Tollis},
  title        = {A Unified Approach a Visibility Representation of Planar Graphs},
  journal      = {Discret. Comput. Geom.},
  volume       = {1},
  pages        = {321--341},
  year         = {1986},
  doi          = {10.1007/BF02187705},
  bibsource    = {dblp computer science bibliography, https://dblp.org}
}

@article{Whitney31b,
author = {Hassler Whitney },
title = {Non-Separable and Planar Graphs},
journal = {Proceedings of the National Academy of Sciences},
volume = {17},
number = {2},
pages = {125-127},
year = {1931},
doi = {10.1073/pnas.17.2.125},
}
\label{sec:biblio}

\end{document}